\documentclass[conference]{IEEEtran}
\IEEEoverridecommandlockouts
\usepackage{cite}
\usepackage{amsmath,amssymb,amsfonts}
\usepackage{algorithmic}
\usepackage{graphicx}
\usepackage{textcomp}
\usepackage{xcolor}
\usepackage{bm}
\usepackage{wrapfig}
\usepackage{multirow}
\usepackage{multicol}
\usepackage{amsmath}
\usepackage{verbatim}
\usepackage{url}
\usepackage{subfigure}
\usepackage{float}
\usepackage{booktabs}
\usepackage{ntheorem}
\usepackage{mathtools}
\usepackage{amssymb}
\usepackage{calc}
\usepackage{cite}
\usepackage{array}
\usepackage{enumerate}
\usepackage{algorithm}
\usepackage{algorithmic}
\usepackage{psfrag}
\usepackage{xspace} 
\usepackage{xspace,xcolor}

\usepackage{bm}
\usepackage{wrapfig}
\usepackage{multirow}
\usepackage{multicol}
\usepackage{amsmath}
\usepackage{verbatim}
\usepackage{url}
\usepackage{subfigure}
\usepackage{float}
\usepackage{booktabs}
\usepackage{ntheorem}
\usepackage{mathtools}
\usepackage{amssymb}
\usepackage{calc}
\usepackage{cite}
\usepackage{array}
\usepackage{enumerate}
\usepackage{algorithm}
\usepackage{algorithmic}
\usepackage{psfrag}
\usepackage{xspace} 
\usepackage{xspace,xcolor}

\makeatletter

\floatstyle{ruled}
\newfloat{algorithm}{tbp}{loa}
\floatname{algorithm}{Algorithm}

\newcommand{\hiset}{\mbox{$\bm{\mathcal{H}}$}\xspace}
\newcommand{\bset}{\mbox{$\bm{\mathcal{B}}$}\xspace}
\newcommand{\gset}{\mbox{$\bm{\mathcal{G}}$}\xspace}

\newcommand{\tabincell}[2]{\begin{tabular}{@{}#1@{}}#2\end{tabular}}

\providecommand{\tabularnewline}{\\}

\newtheorem{lm}{Lemma}

\newtheorem*{proof}{Proof}

\usepackage{ifpdf}

\usepackage{cite}

\usepackage{amsmath}
\usepackage{algorithmic}

\usepackage{array}
\def\BibTeX{{\rm B\kern-.05em{\sc i\kern-.025em b}\kern-.08em
    T\kern-.1667em\lower.7ex\hbox{E}\kern-.125emX}}

\begin{document}

\title{Privacy-Aware Cost-Effective Scheduling Considering Non-Schedulable Appliances\\
in Smart Home}

\author{Boyang~Li,~\IEEEmembership{Student Member,~IEEE,}
        Jie~Wu,
        and~Yiyu~Shi,~\IEEEmembership{Senior Member,~IEEE}
\thanks{B. Li and Y. Shi are with the Department of Computer Science and Engineering, University of Notre Dame, Notre Dame, IN 46556 USA (e-mail: bli1@nd.edu; yshi4@nd.edu).}%
\thanks{J. Wu was with the Department of Computer Science and Engineering,  University of Notre Dame, IN 46556 USA (e-mail: jane.dream.wu@gmail.com)}

\thanks{This paper is supported in part by a Luksic Family Grant.
}
}

\maketitle
\begin{abstract}
A Smart home provides integrating and electronic information services  to help residential users manage their energy usage and bill cost, but also exposes users to significant privacy risks due to fine-grained information collected by smart meters. Taking account of users' privacy concerns, this paper focuses on cost-effective runtime scheduling designed for schedulable and non-schedulable appliances. To alleviate the influence of operation uncertainties introduced by non-schedulable appliances, we formulate the problem by minimizing the expected sum of electricity cost under the worst privacy situation. Inventing the iterative alternative algorithm, we effectively schedule the appliances and rechargeable battery in a cost-effective way while satisfying users' privacy requirement. Experimental evaluation based on real-world data demonstrates the effectiveness of the proposed algorithm.
\end{abstract}


\section{Introduction}
With the development of Internet of Things (IoT) \cite{xu2018quantization}\cite{xu2018mda}\cite{xu2018scaling}\cite{xu2018accelerating}\cite{xu2017efficient}, smart home, managed by intelligent devices \cite{al2014novel} \cite{al2017demand}, have provided tangible benefits
for customers to control and lower their electricity costs. 
For instance, smart meters, which are used for energy
efficiency, are being aggressively deployed in homes and
businesses as a critical component in smart grids. Each customer can control the energy consumption by shifting the operation of
appliances from high price hours to low price hours in order to reduce electricity expense and the
peak-to-average ratio \cite{palensky2011demand} \cite{yue2011dual}. However, attackers can
manipulate smart devices, generate fake electricity price, and 
identify consumers' personal behavior patterns by monitoring the power consumption peaks
to cause physical, psychological, and
financial harm~\cite{molina2010private} \cite{yuan2011modeling} \cite{liu2016smart}. A straightforward
idea to handle such attacks is to add extra uncertainty in the individual load
information by perturbing the aggregate load
measurement \cite{sankar2013smart} \cite{liu2015impact}. However, this approach has to 
modify the metering infrastructure which might not be practical because
millions of smart meters have already been installed. In addition, applying uncertainty 
into customers' power consumption results in inaccurate billing cost. 
It is thus important to carefully consider customers' privacy in
energy scheduling design for smart homes. 

Recently, several studies have paid attention to customers' privacy in
power consumption scheduling design. 
Tan et al.~\cite{tan2013increasing} proposed a power
consumption scheduling strategy to balance the consumers' privacy and energy
efficiency by using an energy harvesting technique. Kalogridis et al.~\cite{kalogridis2010privacy} proposed a power
mixing algorithm against power load changes by introducing a rechargeable
battery. The goal of the proposed algorithm is to maintain the current load
equal to the previous load. McLaughlin et al.~\cite{mclaughlin2011protecting}
proposed a non-intrusive load leveling (NILL) algorithm to combat potential invasions
of privacy. The proposed NILL algorithm adopted an in-residence battery to
offset the power consumed by appliances, to level the load profile to a
constant target load, and to mask the appliance features. 
Chen et al.~\cite{chen2013residential} explored the trade-off between the
electricity payment and electric privacy protection using Monte Carlo
simulation. Yang et al.~\cite{yang2015cost} proposed a scheduling framework for
smart home appliances to minimize electricity cost and protect the privacy of
smart meter data using a rechargeable battery. 
Liu et al.~\cite{liu2016optimal} explored a stochastic gradient method 
to minimize the weighed sum of financial cost and the deviation from the load profile. 
All these existing works assume the activity of every appliance can be precisely scheduled and focus on scheduling algorithm design to trade off between the
electricity cost and the customers' privacy protection for schedulable
household appliances.

However, active operations of household appliances are not always schedulable all the time. They can be classified into
two groups in terms of controllability: schedulable appliances and non-schedulable appliances. 
The operation time of schedulable appliances can be postponed to later time
during the period under consideration, like a laundry and dish washer. 
These appliances can be scheduled and turned on/off by a scheduler. 
The non-schedulable appliances are those that their usages are determined by
the user and are not negotiable, like a TV, laptop and oven. These appliances must
be turned on immediately upon the users' request and cannot be scheduled ahead
of time. Non-schedulable appliances introduce operation uncertainties. 
None of the existing work have considered the influence of  
the uncertainties of non-schedulable appliances on customers' privacy and the corresponding
customers' comfort. 
Our previous work Wu et al. \cite{wu2016privacy} \cite{wu2015efficient}
formulated the scheduling problem by minimizing the expected sum of electricity
cost and achieving acceptable privacy protection.

To further consider, we provide a runtime scheduling framework to comprehensively consider the impact of non-schedulable appliances on customers' privacy and their comfort. 
To our best knowledge, this is the first work addressing 
non-schedulable appliances comprehensively.
The proposed framework adopts a novel iterative algorithm for
efficient characterization of non-schedulable appliances' effects. It optimizes
the electricity costs by incorporating customers' privacy for both schedulable
and non-schedulable appliances as well as the rechargeable battery. The
proposed algorithm is evaluated using real-world household data. 
The results demonstrate that the design of non-schedulable module and runtime scheduling framework will 
propose an operation solution of schedulable appliances to prepare the worst privacy scenarios, 
minimize electricity cost, and provide privacy protection guarantee
when non-schedulable appliances operate in any situation.

\section{Background and Motivation}\label{sec:model}

This section presents an overviews of the smart home system model and introduces the impact of
non-schedulable appliances on customers' privacy and their comfort.

\subsection{Smart Home System Overview}
We use a well-known smart home system as discussed
in~\cite{wu2016privacy}. Both schedulable and non-schedulable appliances as
well as energy storage device, like rechargeable batteries, receive power from
utility provider via a smart meter device and managed by a power management
unit (PMU). In this paper, we adopt the same load models, rechargeable battery
model, customers' privacy model, and price model in~\cite{wu2016privacy}. 
The models are briefly summarized below.

\subsubsection{Load model for appliances}

\noindent\textbf{Load model for schedulable appliances:} 

Let $p_{i}$ be the average power consumption of $i$-th appliance ($i=1,2,\cdots,N$, 
where $N$ is the total number of schedulable appliances).
$r_{i}(t)$ and $r_{i}(t+1)$ be the remaining operation duration of $i$-th appliance at time slot
$t$ and $t+1$, respectively. 
We denote $y(t)$ as the total energy consumption of appliances at each
time slot $t$ ($t=1,2,\cdots,\tau$) and $\tau$ is a scheduling horizon denoted by y(t) can be obtained as
\begin{equation}
  \label{eqn:y1}
  \scalebox{0.8}{%
  $y(t)=\sum_{i=1}^{N}p_{i}\cdot(r_{i}(t)-r_{i}(t+1)) \forall i\in\{1,2,\cdots,N\}$},
\end{equation}
The value of $r_{i}(t)$ and $r_{i}(t+1)$ satisfies 0 $\leq$ $r_{i}(t)$ $\leq$ $r_{i}(t+1)$ $\leq$ $E_{i}/p_{i}$. $E_{i}$ is the known work load of $i$-th appliance.
$r_{i}(t+1)$ can be computed as 
\begin{equation}
  \label{eqn:r2}
  \scalebox{0.8}{%
  $r_{i}(t+1)=\max\left(r_{i}(t)-\sum_{j=1}^{t}x_{i}(j),0\right)$},
\end{equation}
where $r_{i}(1)=E_{i}/p_{i}$ is the initial $i$-th appliance state. 
\begin{equation}
  \label{eqn:constraint}
  \scalebox{0.8}{%
  $\sum_{t=1}^{\tau}x_{i}(t)=1, \forall i\in\{1,2,\cdots,N\}$,}
\end{equation}
where $x_{i}(t)$ is a binary variable. If $x_i(t) = 1$, the $i$-th schedulable
appliance start its operation at time slot $t$. Otherwise, $x_i(t) = 0$. 

\noindent\textbf{Load model for non-schedulable appliances:}
Since the duration of a non-schedulable appliance is unknown, we use a single parameter, $w(t)$ to represent the total power consumption of all operating
non-schedulable appliances at time slot $t$.

\subsubsection{Rechargeable Battery Model}
Let $B(t)$ be the battery state at time slot $t$, which is a function of
battery charge/discharge power ($z(t)$). The battery state at time slot $t+1$ can be expressed as 
\begin{equation}
  \label{eqn:B2}
    \scalebox{0.8}{%
  $B(t+1)=B(t)+z(t),
  t=1,2,\cdots.$}
\end{equation}
where $B(1)$ is the initial battery state. $B(t)$ must satisfy
\begin{equation}
  \label{eqn:B3}
  \scalebox{0.8}{%
  $0 \leq B(t) \leq B_{\max},
  t=1,2,\cdots,\tau .$}
\end{equation}
where $B_{\max}$ is the maximum battery capacity.  
From the perspective of the PMU, it schedules the action variable
$z(t)$ and decides how much power should be charged to or discharged from the battery.

\subsubsection{Customers' Privacy Model}
Let $l(t)$ denote the total aggregated load over the scheduling horizon
$\tau$  and $l(t)$ can be computed by 
\begin{eqnarray}
  \label{eqn:l}
\scalebox{0.8}
 l(t)=&&y(t)+z(t)+w(t)\nonumber\\
  =&&\sum_{i=1}^{N}p_{i}\min \left(\sum_{j=1}^{t}x_{i}(j), r_i(t)\right)+z(t)+w(t)\nonumber\\
  =&&\mathbf{p}^{T}\min \left(\mathbf{V}(t),\mathbf{R}(t)\right)+w(t),
\end{eqnarray}
where $\footnotesize\mathbf{V}(t)=\left[\sum_{j=1}^{t}x_{1}(j),\cdots,\sum_{j=1}^{t}x_{N}(j),z(t)\right]^{T}$ is a $(N+1)\times1$ vector; 
$\footnotesize\mathbf{p}=\left[p_{1},p_{2},\cdots,p_{N},1\right]^{T}$ is a $(N+1)\times1$ vector;
$\mathbf{R}(t)=\left[r_{1}(t), r_{2}(2),\cdots,r_{N}(t),+\infty,\right]^{T}$ is a $(N+1)\times1$ vector;
$[\cdot]^{T}$ is the transpose operation.

If the aggregated load profile is known, the customers' privacy information, such as in-house activatiate, can be obtained. 
To mask such privacy information, we use $\lambda$ to flatten aggregated load profile.
Using the concept of running average historical load ($\overline{l}$), which is
defined as $\overline{l}=\frac{1}{\tau}\sum_{t=1}^{\tau}l(t)$, the customers'
privacy requirement is described as
\begin{equation}
  \label{eqn:privacy}
      \scalebox{0.8}{%
  $-\lambda\leq l(t)-\overline{l}\leq\lambda, \vspace*{-1ex}$}
\end{equation}
where $\lambda\geq0$ is a bounding parameter used to guarantee the privacy.
The larger the value of $\lambda$, the more flexible it is for the customers' privacy
requirement.

\subsubsection{Price Model}
Let $c(t)$ be the per-unit electricity cost received from the utility provider at time slot $t$. 
The time-varying price value of $c(t)$ follows the price discussed in\cite{wu2016privacy}.

\subsection{Impacts of Non-Schedulable Appliances}
To analyze the impacts of non-schedulable appliances, we consider a simple
example with two schedulable appliances ($\alpha_1$ and $\alpha_2$), e.g. laundry and one non-schedulable appliance ($\beta$), e.g. a TV. Let scheduling
horizon $\tau=4$ and privacy bound $\lambda=40$. 
For $\alpha_1$ and $\alpha_2$, we set $E_{\alpha_1}=60$ kW,
$E_{\alpha_2}=80$ kW, $p_{\alpha_{1}}=40$ kW, and $p_{\alpha_{2}}=30$ kW. 
Fig.~\ref{fig:sample1} illustrates the scheduled operations of the two schedulable
appliances, household electricity usage, and privacy ($l(t)-\overline{l}$) without considering
non-schedulable appliances.
\begin{figure}[t]
  \begin{centering}
	\includegraphics[width=0.42\textwidth]{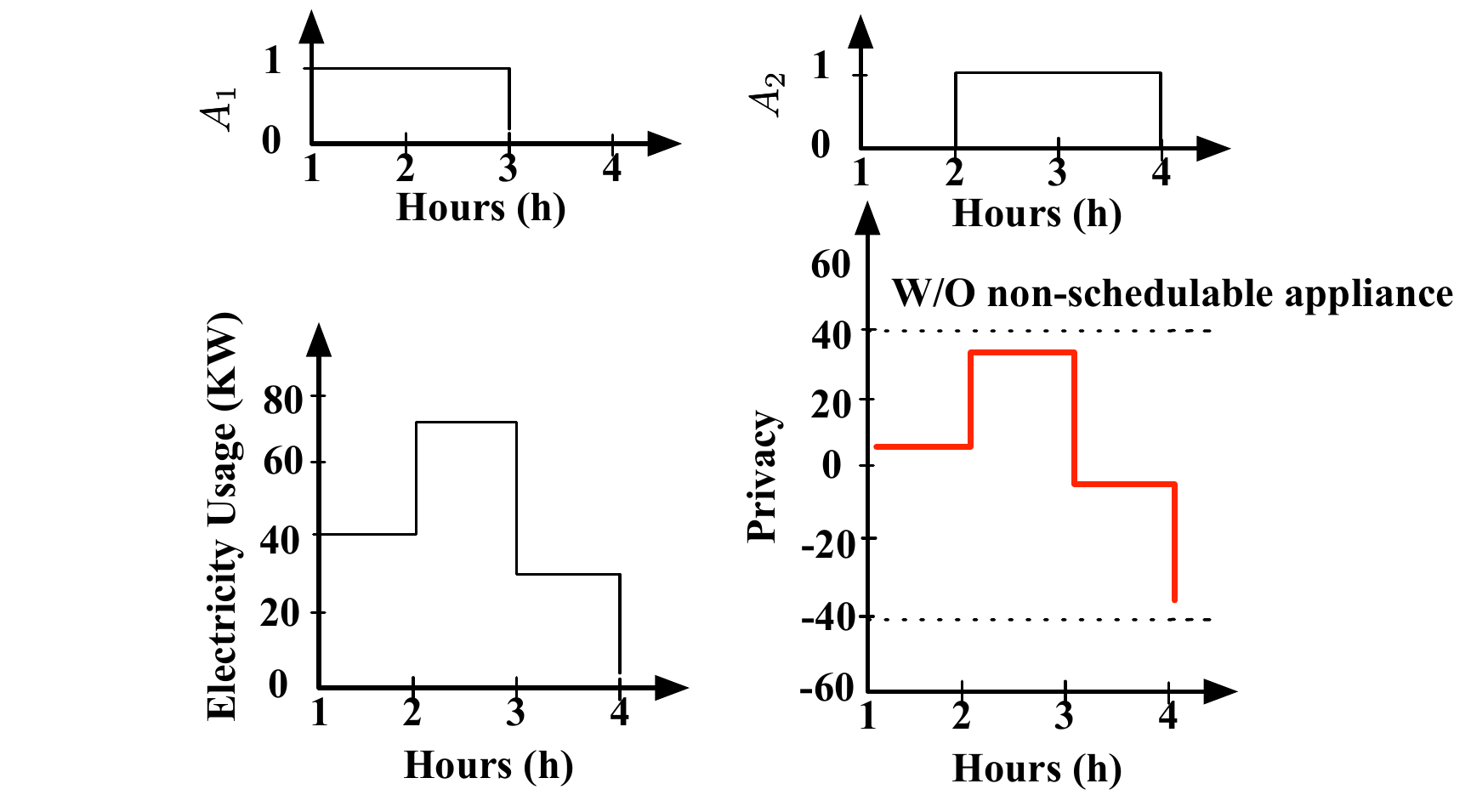}
	\par
  \end{centering}
  \vspace*{-3ex}
  \caption{The Operations of $\alpha_1$ and $\alpha_2$, household electricity usage, and privacy.}
  \label{fig:sample1} \vspace*{-3ex}
\end{figure}

In real-world, customers may turn on the non-schedulable appliance at anytime. 
We assume that the non-schedulable appliance
$\beta$ operates at time slot $t=2$ and $t=3$ with a high possibility based
on customers' historical behaviors. Fig.~\ref{fig:sample2} illustrates the
household electricity usage and privacy including the non-schedulable appliance
operation. 
\begin{figure}[t]
  \begin{centering}
	\includegraphics[width=0.42\textwidth]{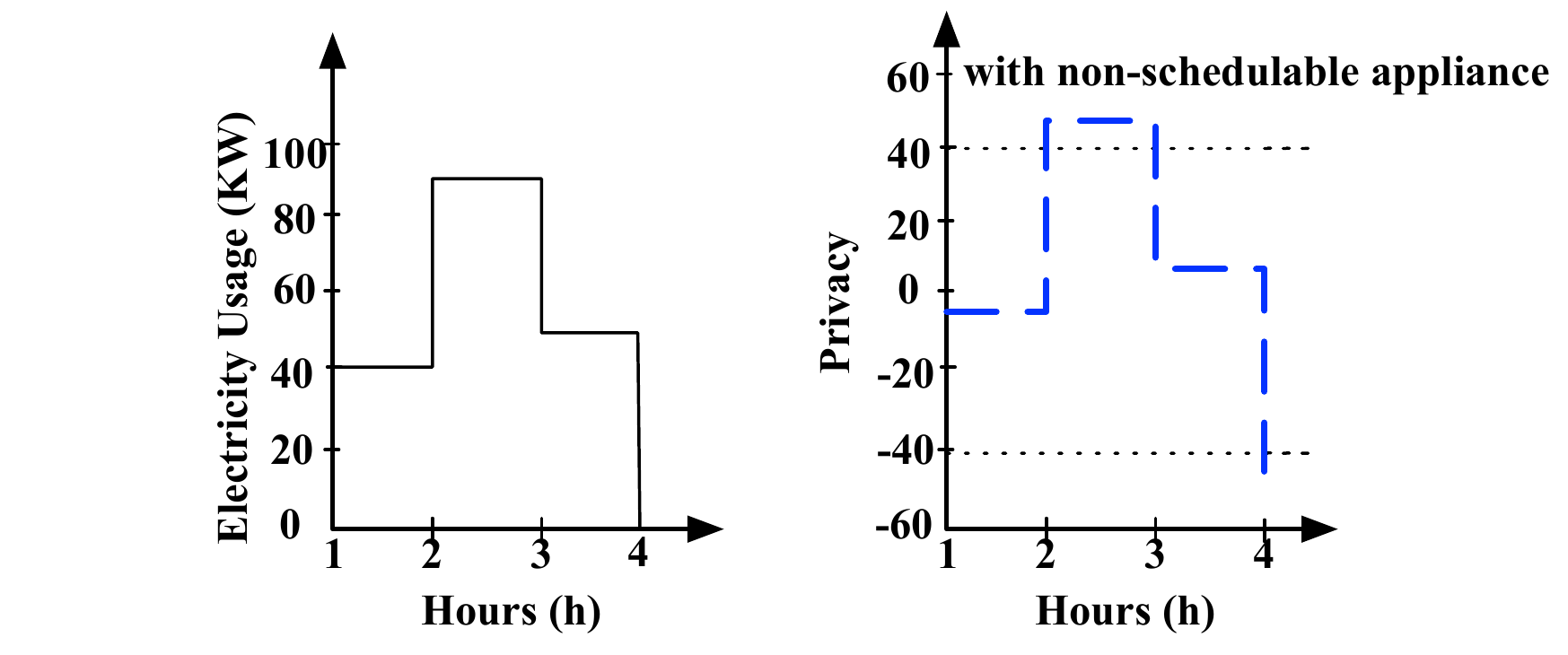}
	\par
  \end{centering}
  \vspace*{-3ex}
  \caption{The household electricity usage and privacy within a non-schedulable
  appliance.}
  \label{fig:sample2} \vspace*{-3ex}
\end{figure}
Comparing Fig.~\ref{fig:sample1} and Fig.~\ref{fig:sample2}, one can see that the privacy
breach occurs when a non-schedulable appliance operates in the scheduling horizon. Fig.~\ref{fig:sample1} and Fig.~\ref{fig:sample2} also
demonstrates the trade-off between customers' privacy and contentment requirements. 
In order to fit the privacy constrain, it is necessary to re-schedule the runtime of schedulable appliances and charging/discharging time of battery.
Therefore, to
ensure the customers' privacy protection, understanding of the
influence of non-schedulable appliance becomes essential. 

\section{Runtime Appliance Scheduler}\label{sec:design3}

In this section, we present
our approach to runtime cost-effective appliance scheduling to satisfy both privacy and contentment within considering
the effect of non-schedulable appliance. Specifically, Section~\ref{subsec:scheduling-overview} describes
the hierarchical structure of Runtime Appliance Scheduler, RAS. 
Section~\ref{subsec:schedulable-problem} introduces
the specific scheduling problem to be be solved. Section~\ref{subsec:algorithm} discusses
the algorithm design for solving this problem.

\subsection{RAS Overview}
\label{subsec:scheduling-overview}
The overall structure of RAS is shown in Fig.~\ref{fig:scheduling-overview}.
The input components of the RAS is the
schedule table and the current operation states of both schedulable and
non-schedulable appliances as well as the battery state. Given these two input
components, RAS can easily lookup the schedule table
and find the optimal scheduling solution in real time. Among these input components,
building the schedule table at the beginning of scheduling horizon plays the
key role in RAS design.  

\begin{figure}[t]
  \begin{centering}
	\includegraphics[width=0.42\textwidth]{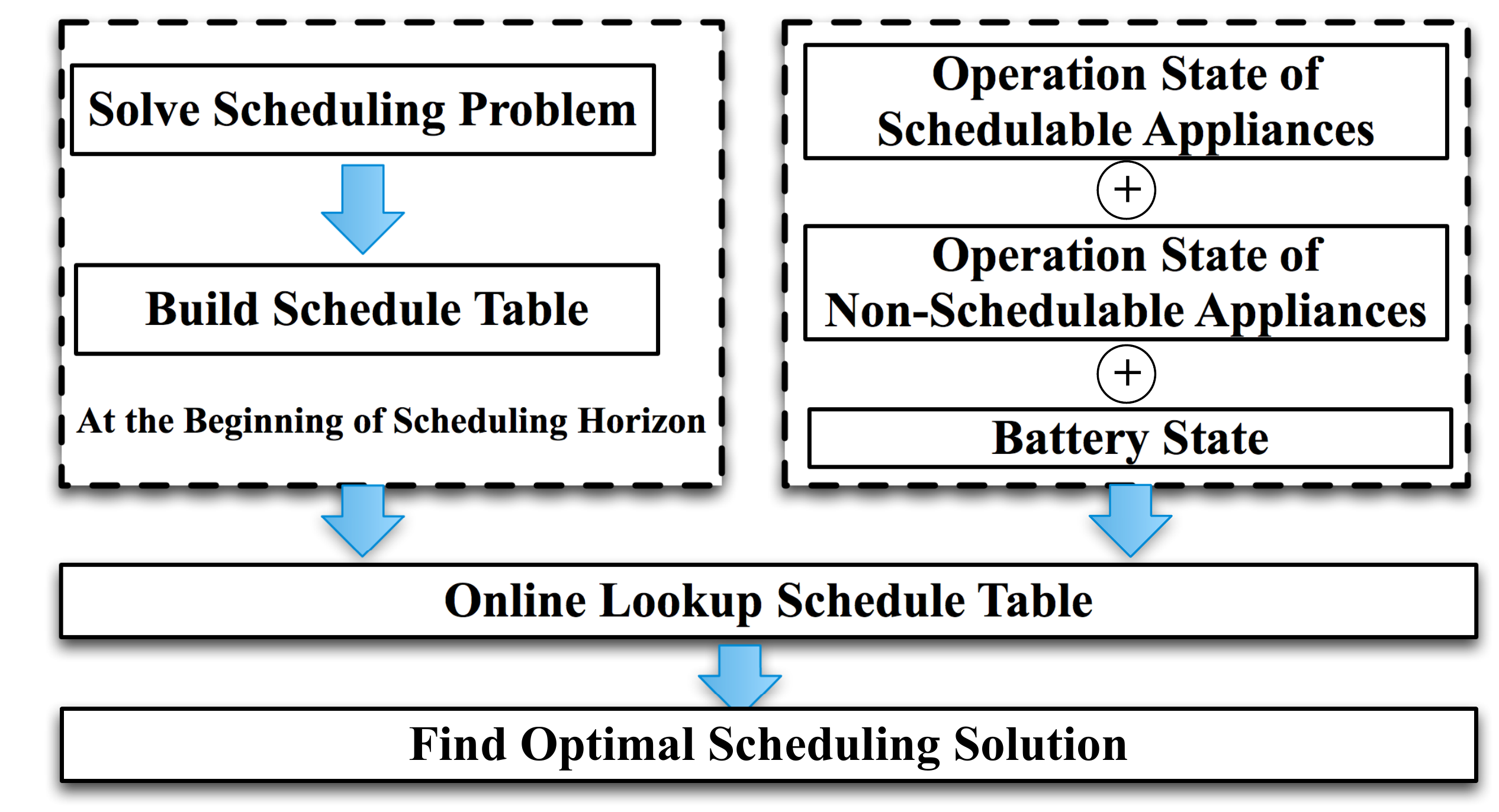}
	\par
  \end{centering}
  \vspace*{-2ex}
  \caption{The overall structure of the runtime appliance scheduler.}
  \label{fig:scheduling-overview} \vspace*{-3ex}
\end{figure}

\subsection{Scheduling Problem Formulation}
\label{subsec:schedulable-problem}
The power management unit (PMU), utilizing models in Section~\ref{sec:model}, is designed to schedule the battery and appliances operation. $x_{i}(t)$ ($i=\{1,2,\cdots,N\}$) for all schedulable appliances and $z(t)$ for the rechargeable battery. This optimization problem ($\sf{SP}$) can be expressed as
\begin{subequations}
\label{eqn:schedule}
\vspace*{-2ex}
\begin{alignat}{2}
\scalebox{0.8}
\sf{SP}:~&{\mbox{min.}}\;&&{\sf{E}}\{\sum_{t=1}^{\tau}C(\mathbf{V}(t))\} =
{\sf{E}}\{\sum_{t=1}^{\tau} \left(c(t)l(t)\right)\}\; \nonumber \\
&\quad&& ={\sf{E}}\{\sum_{t=1}^{\tau}c(t)\left(\mathbf{p}^{T}\min
\left(\mathbf{V}(t),\mathbf{R}(t)\right)+w(t)\right)\}\\
&\mbox{s.t.}
&& -\lambda\leq\left(\mathbf{p}^{T}\min
\left(\mathbf{V}(t),\mathbf{R}(t)\right)+w(t)\right)-\overline{l}\leq\lambda,\label{eqn:constraint1}\\
&\quad && r_{i}(t+1)=\max\left(r_{i}(t)-\sum_{j=1}^{t}x_{i}(j),0\right),\\
&\quad && \sum_{t=1}^{\tau}{x_{i}(t)} = 1, x_{i}(t)\in\{0,1\}, \forall i \in \{1, 2, \cdots, N\},\\
&\quad && B(t+1)=B(t)+z(t), 0\leq B(t) \leq B_{\max}\\
&\quad && z_{\min}\leq z(t)\leq z_{\max}\\
&\quad && t=1,\cdots,\tau. \nonumber
\end{alignat}
\end{subequations}
where $C(\mathbf{V}(t))=c(t)\left(\mathbf{p}^{T}\min \left(\mathbf{V}(t),\mathbf{R}(t)\right)+w(t)\right)$.
$z_{\min}$ is the maximum discharged power, which is also the minimum charged power; $z_{\max}$ is the maximum
charged power. $\sf{E}\{\cdot\}$ represents the expectation function. 
$\mathbf{V}(t)$ is the vector of decision variables for scheduling both the
schedulable appliances and the rechargeable battery over the scheduling horizon. 
The expectation function is needed in the objective function above due to
operation uncertainties of non-schedulable appliances. 

Balancing electricity
cost within a scheduling horizon in the presence of uncertainties of power consumption by non-schedulable appliances
is thus a dynamic process. However, solving the optimization problem
$(\sf{SP})$ at each time slot would be too time consuming. To address this
issue, we present a hybrid approach. 

\subsection{Scheduler Algorithm Design} 
\label{subsec:algorithm}
In this subsection,  we introduce a hybrid approach to solve the optimization
problem $\sf{SP}$ in~\eqref{eqn:schedule}. 
In the problem $\sf{SP}$, the uncertainty effects introduced by
non-schedulable appliances need to be carefully considered in the privacy
constraint~\eqref{eqn:constraint1}, Because the immediately active operations by
non-schedulable appliances lead to high peak load profiles and leak the
appliance features. Lacking a comprehensive consideration of the influence of
non-schedulable appliances can leak customers' privacy. To handle this issue, we consider the
worst influences of non-schedulable appliances' operation in the 
customers' privacy constraint~\eqref{eqn:constraint1}. We define a time zone
of a peak load profile over the scheduling horizon as a worst privacy
scenario. Mathematically, we define $\varphi=[t_{l},t_{u}]$ as a time zone for worst privacy scenario.
$t_{l}$ and $t_{u}$ are a lower bound and a upper bound of time slot, separately. 
To guarantee the scheduler satisfies the
privacy constraint~\eqref{eqn:constraint1} even when the non-schedulable appliances operate in the
worst scenario, we present a hybrid approach to handle it. 

First, we assume that non-schedulable appliances are active at the worst privacy
scenario. Once the non-schedulable appliances are assigned, we apply a dynamic
programming like algorithm to solve the optimization problem $\sf{SP}$ and
build a schedule table at the beginning of each scheduling horizon.
This table contains $\tau$ columns where $\tau$ is the total number of time slots in the scheduling horizon, and a number of rows corresponding to the different states (described by the $r_i(t)$ and $B(t)$ values).
Each entry in the table for time slot $t$ and state $s$ contains the assignment to $x_i(t)$ and $z(t)$ for a given set of $r_i(t)$ and $B(t)$ values.
Note that we determine the assignment of $x_i(t)$ and $z(t)$ by consulting the
table entry for time $t$ and state $s$ at the beginning of each time slot
$t$. Second, we iteratively update worst privacy scenario after each 
schedule table building. Once the worst privacy scenario does not change
anymore, we collect all these privacy scenarios as a potential set of operation time
zone for non-schedulable appliances. And then, we re-apply a dynamic
programming like algorithm to find the optimal assignment for schedulable
appliances and rechargeable battery.  
This hybrid approach effectively takes into consideration of both non-schedulable and schedulable appliances. 

\subsubsection{A Dynamic Programming Algorithm Design}
Given the operation assignments of non-schedulable appliances at the worst privacy
scenario, the schedule table can be obtained by solving optimization problem
$\sf{SP}$. To ensure that the size of a schedule table is manageable, we assume that the remaining
operation duration $r_{i}(t)$ ($\forall i$) and the battery state $B(t)$
are generally discretized to finite sets:
$r_{i}(t)\in \hiset =(\mathcal{H}_{i,1},\cdots,\mathcal{H}_{i,S_{i}})$ where $S_{i}=E_{i}/p_{i}$ ($\forall i$), and $B(t)\in \bset =(\mathcal{B}_{1},\cdots,\mathcal{B}_{M})$.
Here \hiset is the $S$-element remaining operation duration
set for the $i$-th schedulable appliance, and \bset is the $M$-element battery state set.
Let \gset be the state set with $M\times\prod_{i=1}^{N}S_{i}$ elements including battery state set \bset and remaining operation duration set \hiset ($\forall i$). 
Then, the structure of a schedule table can be shown as Table~\ref{tab:sched-table}. 
That is, a schedule table consists of $\tau$ sub-tables (corresponding to the
columns in Table~\ref{tab:sched-table}), each of which is for a specific time
slot $t$ from $1$ to $\tau$. Each sub-table consists of
${M\times\prod_{i=1}^{N}S_{i}}$ entries (corresponding to the state set), each
of which contains the assignment to $\mathbf{x}(t) = [x_1(t), x_2(t), \cdots,
x_N(t)]^T$. 

The formulation~\eqref{eqn:schedule} 
 in Section~\ref{sec:design3} forms the basis for constructing the schedule table.
Specifically, we adopt a backward recursive approach to solve problem
$\sf{SP}$. Given the initial state $\mathcal{G}_1$
at time slot $1$, ($\mathcal{G}_1\in \gset$), we denote
$\mathcal{F}_{1}(\mathcal{G}_1)$ as the optimal value of~\eqref{eqn:schedule},
which can be obtained recursively due to the principle of
optimality~\cite{bertsekas1995dynamic}. Because the value of $C(\mathbf{V}(t))$
can be precisely determined since the operation states of both schedulable and
non-schedulable appliances as well as the battery states prior to $t$ are
known. Only the future states are not known. Thus we can rewrite~\eqref{eqn:schedule} as
\begin{subequations}
\label{eqn:schedule_recursive}
\vspace*{-1ex}
\begin{alignat}{2}
\scalebox{0.8}
\sf{SP}:& && \mathcal{F}_{1}(\mathcal{G}_1) = 
\min \left\{C(\mathbf{V}(1)) + \sf{E}_{\gset} \left \{\sum_{t=2}^\tau
C(\mathbf{V}(t))\right\}\right\} \\
&\mbox{s.t.} \quad
&& -\lambda\leq\left(\mathbf{p}^{T}\min
\left(\mathbf{V}(t),\mathbf{R}(t)\right)+ w(t)\right)
-\overline{l}\leq\lambda,\\
&\quad &&r_{i}(t+1)=\lceil\max(r_{i}(t)-\sum_{j=1}^{t}x_{i}(j),0)\rceil\nonumber\\
&\quad &&  r_{i}(t),r_{i}(t+1)\in \hiset \label{eqn:rceil}\\
&\quad && B(t+1)=B(t)+z(t),\nonumber\\
&\quad && B(t),B(t+1)\in \bset, \label{eqn:bfloor}\\
&\quad && \sum_{t=1}^{\tau}{x_{i}(t)} = 1, x_{i}(t)\in\{0,1\},\\
&\quad && z_{\min}\leq z(t)\leq z_{\max}\\
&\quad && t=\tau-1,\tau-2,\cdots,1, \hspace*{1em} i=1,2,\cdots,N \nonumber
\vspace*{-3ex}
\end{alignat}
\end{subequations}
where $\mathbf{V}(1)$ corresponds to those $x_{i}(1)$ and $z(1)$ that result in state $\mathcal{G}_1$ at time $t=1$.
$\sf{E}_{\gset}$ describes the expectation operation over the all possible
states $\gset$.
Furthermore,~\eqref{eqn:rceil} denotes the state transition of remaining operation
duration for all schedulable appliances, which is constrained in the finite set \hiset. 
The \eqref{eqn:bfloor} denotes the battery state transition that is also constrained in the finite set \bset.

In a nutshell, the optimal values for arbitrary time slots $(t)$ to $1$ are
determined in a backward recursive manner by considering state transitions from
all possible state $\mathcal{G}_{t+1}$ at $t+1$ to $\mathcal{G}_{t}$ at $t$
($\mathcal{G}_{t},~\mathcal{G}_{t+1}\in\gset$) and the
constraints in \eqref{eqn:schedule_recursive}, which is shown as follows.
\begin{equation}
  \label{eqn:obj-recursive}
  \scalebox{0.8}{}{%
  \mathcal{F}_{t}(\mathcal{G}_{t})=\min\{C(\mathbf{V}(t))+\sf{E}_{\gset}\{\mathcal{F}_{t+1}(\mathcal{G}_{t+1})\},}  
\end{equation}
where $C(\mathbf{V}(t))$ is the electricity cost value in state
$\mathcal{G}_{t}$ at time slot $t$, which corresponds to those $x_{i}(t)$ and
$z(t)$. $\mathcal{G}_{t}$ and $\mathcal{G}_{t+1}$ are the state at time slot
$t$ and $t+1$, respectively. $\sf{E}_{\gset}\{\mathcal{F}_{t+1}(\mathcal{G}_{t+1})\}$ is the
expected sum of the minimal cost value over all possible states $\gset$ for time slots
$t+1,t+2,\cdots,\tau$.
Note that, the backward recursive approach firstly calculates the optimal value
$\mathcal{F}_{\tau}(\mathcal{G}_{\tau})$ for time slot $\tau$ with the known state
$\mathcal{G}_\tau$, which is shown as follow.
\begin{equation}
  \label{eqn:tau-value}
  \scalebox{0.8}{}{%
  \mathcal{F}_{\tau}(\mathcal{G}_{\tau})=\min C(\mathbf{V(\tau)}).}
\end{equation}

These processes are summarized in Algorithm 1.

\begin{table}[!t]
  \caption{\label{tab:sched-table} Structure of the schedule table.}
  \centering
  \begin{tabular}{c||c||c||c||c}
  \hline 
  \multirow{2}{*}{\tabincell{c}{State\\($B(t), r_{1}(t), r_{2}(t), \cdots, r_{N}(t)$)}} & \multicolumn{4}{c}{Time Slot} \tabularnewline
  \cline{2-5}  \cline{2-5}
  & 1 & 2 & $\cdots$ & $\tau$\tabularnewline
  \hline 
  $\mathcal{B}_{1}, \mathcal{H}_{1,1}, \mathcal{H}_{2,1}, \cdots, \mathcal{H}_{N,1}$ & $\mathbf{x}(1)$ & $\mathbf{x}(2)$&$\cdots$ & $\mathbf{x}(\tau)$\tabularnewline
  \hline 
  $\mathcal{B}_{2}, \mathcal{H}_{1,1}, \mathcal{H}_{2,1}, \cdots, \mathcal{H}_{N,1}$ &$\mathbf{x}(1)$&$\cdots$ & $\mathbf{x}(t)$\tabularnewline
  \hline 
  $\cdots$ & $\cdots$ & $\cdots$& $\cdots$& $\cdots$\tabularnewline
  \hline 
  $\mathcal{B}_{M}, \mathcal{H}_{1,S_1}, \mathcal{H}_{2,S_2}, \cdots, \mathcal{H}_{N,S_N}$ & $\mathbf{x}(1)$ & $\mathbf{x}(2)$& $\cdots$ & $\mathbf{x}(\tau)$\tabularnewline
  \hline 
  \end{tabular} \vspace*{-2ex}
\end{table}

\begin{algorithm}
  \begin{algorithmic}[1]
    \STATE {\bf Initialization}
    \STATE \quad Given initial state $\mathcal{G}_{1}$, 
    $M$-element battery state set $\bset$, $S$-element remaining
    operation duration set $\hiset$ for $i$-th schedulable appliance ($\forall
    i$), number of time slots $\tau$, and the assignment active operations of non-schedulable appliances
    \STATE {\bf Recursive calculation}
    \STATE \quad {\bf for} $t=\tau$ to $1$
    \STATE \quad\quad {\bf for} $\mathcal{G}_{t}=\mathcal{G}_{1} $ to
    $\mathcal{G}_{M\times\prod_{i=1}^{N}S_{i}}$
    \STATE \quad\quad\quad {\bf if} $t=\tau$, Calculate
    $\mathcal{F}_{\tau}(\mathcal{G}_\tau)$ by \eqref{eqn:tau-value};
    \STATE \quad\quad\quad {\bf else}, Calculate
    $\mathcal{F}_{t}(\mathcal{G}_t)$
    by solving \eqref{eqn:obj-recursive} until
    $\mathcal{F}_{1}(\mathcal{G}_1)$ is obtained.
    \STATE {\bf Return} Optimal value $\mathcal{F}_{1}(\mathcal{G}_{1})$,
    and its corresponding $x_{i}(t)$ ($\forall i$) and $z(t)$ for each time
    slot $t$ ($\forall t\in\tau$). 
  \end{algorithmic}
  \caption{\hspace{-3pt}: A Dynamic Programming Algorithm Design}
  \label{alg:DBDP}
\end{algorithm}

\subsubsection{Privacy-Aware Scheduler Design}
In the description of Algorithm 1, the assignment of active operations of
non-schedulable appliances are known before the dynamic process. 
Considering multiple worst privacy scenarios, we iteratively 
assign non-schedulable appliances into worst scenario and
generate it as a new privacy constraint, and then use assignment active
operations of non-schedulable appliances. Finally, the problem $\sf{SP}$
in~\eqref{eqn:schedule_recursive} is solved. To simplify the
problem, we assume that the power consumption of any non-schedulable appliance
within each time slot during the appliance's active operation is a constant
value, defined as $p_{j}$ ($\forall j=1,2,\cdots, W$). $W$ is the total number
of non-schedulable appliances.
Associated with the set of time zone for worst privacy scenarios ($\Omega$), we
define a binary variable $h_{j}(t,\varphi)$. If $h_{j}(t,\varphi)=1$, the
$j$-th non-schedulable appliance actives its operation at time slot $t$ in
worst privacy scenario $\varphi$.
Otherwise $h_{j}(t,\varphi)=0$. The definition of $h_{j}(t,\varphi)$ is shown
as follows.
\vspace*{-2ex}
\begin{equation}
  \label{eqn:h}
  \scalebox{0.8}{}
    \begin{aligned}
      h_j(t,\varphi) &=
      \begin{cases}
        1 & \text{if } t\in\varphi~(\varphi\in\Omega),\\ 
      0 & \text{else } \text{Otherwise},
     \end{cases}
   \end{aligned}
   \forall j=1,2,\cdots,M.
\end{equation}
Thus, we re-write the load model for non-schedulable appliances. 
The total power consumptions of all non-schedulable appliances at time slot $t$ is expressed as
\vspace*{-5ex}
\begin{equation}
  \label{eqn:w}
  \scalebox{0.8}{}{%
  w(t,\varphi)=\sum_{j=1}^{W}h_j(t,\varphi)p_{j}.}
  \vspace*{-5ex}
\end{equation}
According to~\eqref{eqn:l} and~\eqref{eqn:w}, the updated privacy model can be re-written as 
\vspace*{-7ex}
\begin{equation}
  \label{eqn:privacy2}
  \scalebox{0.8}{}{%
  -\lambda\leq g\left(\mathbf{V}(t),\varphi\right)\leq\lambda,~\forall\varphi\in\Omega, }
  \vspace*{-5ex}
\end{equation}
where $g\left(\mathbf{V}(t),\varphi\right)=\mathbf{p}^{T}\min
\left(\mathbf{V}(t),\mathbf{R}(t)\right)+w(t,\varphi)-\overline{l}$.
Taking account of the set of worst privacy scenarios ($\Omega$) and updated privacy
constraint~\eqref{eqn:privacy2}, the scheduling problem $(\sf{SP}(\Omega))$ is expressed as
\begin{subequations}
\label{eqn:schedule2}
\vspace*{-3ex}
\begin{alignat}{2}
\small
\scalebox{0.8}{}
&\sf{SP}(\Omega):&&
\mathcal{F}_{1}(\mathcal{G}_1) = \min \left\{C(\mathbf{V}(1)) + \sf{E}_{\gset} \left \{\sum_{t=2}^\tau
C(\mathbf{V}(t))\right\}\right\} \\
&\mbox{s.t.}\quad
&& -\lambda\leq
g\left(\mathbf{V}(t),\varphi\right)\leq\lambda,~\forall\varphi\in\Omega\label{eqn:constraint2}\\
&\quad &&
r_{i}(t+1)=\max\left(r_{i}(t)-\sum_{j=1}^{t}x_{i}(j),0\right),\label{eqn:c1}\\
&\quad && \sum_{t=1}^{\tau}{x_{i}(t)} = 1, x_{i}(t)\in\{0,1\}, \forall i \in
\{1, 2, \cdots, N\},\label{eqn:c2}\\
&\quad && B(t+1)=B(t)+z(t), 0\leq B(t) \leq B_{\max}\label{eqn:c3}\\
&\quad && z_{\min}\leq z(t)\leq z_{\max}\label{eqn:c4}\\
&\quad && t=1,\cdots,\tau. \nonumber
\end{alignat}
\end{subequations}

In the optimization problem $\sf{SP}(\Omega)$, ~\eqref{eqn:constraint2}
includes the set of worst privacy scenarios $\Omega$, which consists of
infinite constraints.
Inspired by semi-infinite programming technique, we introduce a hybrid approach
to solve the problem $\sf{SP}(\Omega)$, which is summarized in Algorithm 2.
In Algorithm 2, we denote $\mathbf{V}^{*}_{k}$ and $\mathbf{V}^{*}_{k-1}$ as
the feasible solutions of problem $\sf{SP}(\Omega_{k})$ and
$\sf{SP}(\Omega_{k-1})$, respectively.
Note That at the $k$-th iteration of the Algorithm 2, we update the 
privacy constraint~\eqref{eqn:constraint2} using
$\Omega_{k}=\Omega_{k-1}\cup\{\varphi_{k}\}$. And then we solve a subproblem
$\sf{SP}(\Omega_{k})$ with $\varphi_{k}$ satisfying worst privacy scenario requirement.
Meanwhile, the Algorithm 2 converges after several iterations and produces an approximate optimal solution for $\sf{SP}(\Omega)$ if the
worst privacy scenario doesn't change anymore.

\noindent\textbf{Convergence analysis:} 
Let $\mathcal{V}$ denote the feasible region of problem $\sf{SP}$. At the
$k$-th iteration, when constraints~\eqref{eqn:c1},~\eqref{eqn:c2},~\eqref{eqn:c3} and~\eqref{eqn:c4}
are satisfied, the feasible region of problem $\sf{SP}(\Omega_{k})$ with worst
privacy scenarios $\Omega_{k}$ is expressed as follows.
\begin{equation}
  \label{eqn:feasible-region}
  \scalebox{0.8}{}{%
  \mathcal{V}_{k}:=\{\mathbf{V}|-\lambda\leq g\left(\mathbf{V},\varphi\right)\leq\lambda,~\forall\varphi\in\Omega_{k}\}}
\end{equation}
To prove Algorithm 2 converge to an optimal solution when
$k\rightarrow\infty$, the following two lemmas are needed.
\begin{lm}
  \label{def:convergence}
  For each $k\geq1$, if Algorithm 2 does not stop at this iteration,  
  $\mathcal{V}_{k}\subseteq\mathcal{V}_{k-1}$ holds, where
  $\mathcal{V}_{k-1}$ and $\mathcal{V}_{k}$ are the feasible regions of
  optimization problem $\sf{SP}(\Omega_{k-1})$ and $\sf{SP}(\Omega_{k})$,
  respectively.
\end{lm}
\begin{proof}
  By contradiction, suppose this Lemma is false: for each $k\geq1$, if
  Algorithm 2 does not stop at this iteration, the
  $\mathcal{V}_{k}\nsubseteq\mathcal{V}_{k-1}$. 
  This means that a feasible region $\mathcal{V}_{k-1}$
  is a subset of $\mathcal{V}_{k}$. Therefore, $\mathcal{V}_{k-1}$ is also the
  feasible regions of problem $\sf{SP}(\Omega_{k})$, which satisfies privacy
  constraint~\eqref{eqn:constraint2}. Based on the definition of feasible
  region, the following equation~\eqref{eqn:contradiction} holds.
  \begin{equation}
    \label{eqn:contradiction}
    \scalebox{0.8}{}{
    \mathcal{V}_{k-1}:=\{\mathbf{V}|-\lambda\leq
    g\left(\mathbf{V},\varphi\right)\leq\lambda,~\forall\varphi\in\Omega_{k}\}}
  \end{equation}
  Meanwhile, $\mathcal{V}_{k-1}$ is also the feasible regions of problem
  $\sf{SP}(\Omega_{k-1})$. Therefore, $\mathcal{V}_{k}\nsubseteq\mathcal{V}_{k-1}$ yields $\Omega_{k}=\Omega_{k-1}$.
  However, $\Omega_{k}:=\Omega_{k-1}\cup\{\varphi_{k}\}$ holds and
  $\varphi_{k}$ exists, because Algorithm 2 does not stop
  at $k$-th iteration, contradicting our assumption. Thus, this lemma holds.
\end{proof}
  
\begin{lm}
  For $k\geq1$, if $k\rightarrow\infty$, the subproblem $\sf{SP}(\Omega_{k})$ has an optimal
  solution.
\end{lm}

\begin{proof} When $k\rightarrow\infty$, the set of worst privacy scenarios $\Omega_{k}\rightarrow\Phi$, where
  $\Phi$ is the set of all possible worst privacy scenarios. Once
  $\Omega_{k}$ reaches $\Phi$, the feasible region $\mathcal{V}_{k}$ of $\sf{SP}(\Omega_{k})$
  becomes a unique and smallest feasible region due to lemma 1. 
  Thus, the subproblem $\sf{SP}(\Omega_{k})$ has an optimal solution.
\end{proof}


\begin{algorithm}
  \begin{algorithmic}[1]
    \STATE {\bf Initialization}
    \STATE \quad Apply the Algorithm 1 to solve the problem $\sf{SP}$ without considering
    non-schedulable appliances
    \STATE \quad Find the worst privacy scenario $\varphi_{0}$. Set $\Omega_{0}:=\{\varphi_{0}\}$ 
    \STATE {\bf Iterative programming calculation}
    \STATE \quad Apply the Algorithm 1 to solve $\sf{SP}(\Omega_{0})$ and
    obtain an optimal scheduling variable $\mathbf{V}^{*}_{0}$
    \STATE \quad Set $k:=1$
    \STATE \quad Find the worst privacy scenario $\varphi_{k}$
    \STATE \quad\quad {\bf if} $\varphi_{k}=\varphi_{k-1}$, {\bf then} STOP,
    {\bf return} optimal scheduling variable $\mathbf{V}^{*}_{k}$ 
    \STATE \quad\quad {\bf else}, let $\Omega_{k}:=\Omega_{k-1}\cup\{\varphi_{k}\}$
    \STATE \quad Apply the Algorithm 1 to solve problem $\sf{SP}(\Omega_{k})$ to obtain
    an optimal scheduling variable $\mathbf{V}^{*}_{k}$
    \STATE \quad Set $k:=k+1$ and got to Step 7
  \end{algorithmic}
  \caption{\hspace{-3pt}: Iterative Alternative Algorithm}
  \label{alg:DBDP}
\end{algorithm}

\section{Numerical Simulations}\label{sec:results}

This section evaluates the proposed runtime scheduling framework using
real-world household data. The proposed iterative alternative algorithm
conducts the worst scenario optimization to effectively generate cost-efficient scheduling solution with privacy protection guarantee. 
Section~\ref{subsec:Real world benchmark} demonstrates scheduling results based on real world power consumption. Section~\ref{subsec:schedulerEvaluation} campares the evaluation results of the system without a non-schedulable appliance to the system with a non-schedulable appliance. Section~\ref{subsec:SensiAna} demonstrates the impacts of battery capacity on the scheduler behavior.  

\subsection{Real world benchmark}
\label{subsec:Real world benchmark}

\noindent We first summarize the simulation-based experimental settings. 

\noindent{\bf Appliance data sets and types:} 
We have selected five household appliances data from a ECO data set~\cite{ECO}.
This ECO data set includes aggregate and plug-in appliances' power consumptions of households in Switzerland over
a period of 8 months. These data were collected customer daily usage with $86,400$ measurements per
day. Two types of appliances are considered:
schedulable appliances (i.e., clothes dryer, washing machine, dishwasher, stove, and refrigerator)
and non-schedulable appliances (i.e., PC, stereo, TV, and laptop).
To demonstrate the proposed approach, we consider a household having three
schedulable appliances and two non-schedulable appliances. We set the entire
work load and power consumptions of three schedulable appliances ($1$,
$2$, and $3$) as follows: ($E_{1}=70.7W$,
$p_{1}=35.38Wh$), ($E_{2}=313.2W$, $p_{2}=156.59Wh$), and ($E_{3}=230.2W$,
$p_{3}=76.73Wh$). For the two non-schedulable appliances ($4$ and $5$), the entire work load
$E_{4}=106.97W$ and $E_{5}=33.73W$. Based on the ECO data set~\cite{ECO}, the
appliance $4$ and $5$ usually operates at time slot $t=[7,12]$ and
$t=[1,6]$ with a high possibility, respectively. 
The remaining operation duration set \hiset is discretized as 
\begin{equation}
\scalebox{0.8}{}{%
    \hiset = \{0,1,\cdots,\lceil E_{i}/p_{i}\rceil\}. \nonumber}
\end{equation}
The electricity price is
adopted by the public released data from Ameren Corporation~\cite{price}.  

\noindent{\bf Rechargeable battery parameters:}
The maximum battery capacity $B_{max}=750$ Wh and the initial state of battery
($B(1)$) is 0. To apply the proposed iterative alternative algorithm, we  
discretized the state set \bset of battery as
\begin{equation}
\scalebox{0.8}{}{%
    \bset = \{0,1,2, \cdots, B_{max}\}. \nonumber}
\end{equation}
The battery charged/discharged power set $\mathcal{Z} = \{z|50Wh \leq z \leq 250Wh\}$.
To speed up the table building process, a local version of the algorithm was used in experiment where $z$ was not discretized. 
In terms of privacy concerns, we set $\lambda = 80$.

\noindent{\bf Scheduling horizon:}
The length of the scheduling time slot is one hour. The overall scheduling
horizon is set to be 12 hours. The simulation result is shown in Fig.~\ref{fig:9}. It can be observed that the original power consumption curve exceeded the pravicy constrain and the shape peak is sensitive to be detected by attacker. The scheduled power consumption curve followed the privacy constrain between the upper boundary and lower boundary and the curve is relative smooth.

\begin{figure}[tb!]
  \begin{centering}
	\includegraphics[width=0.42\textwidth]{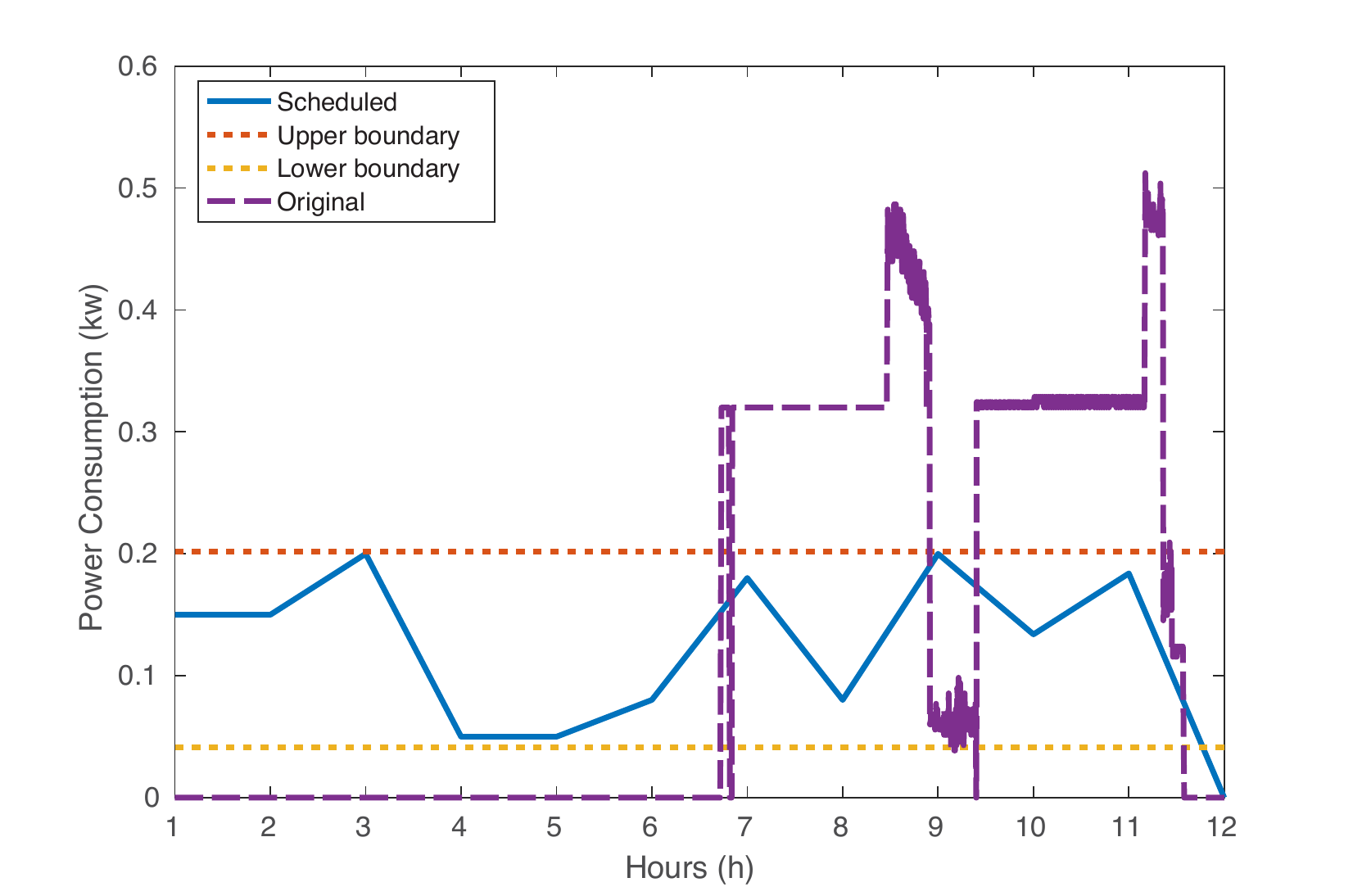}
	\par
  \end{centering}
    \vspace*{-2ex}
  \caption{Real world scheduling result.}
  \label{fig:9}     \vspace*{-4ex}
\end{figure}


\begin{table}[h]
\centering
 \vspace*{-3ex}
\caption{Evaluation environment setting}
\label{tab:EvaluationSetting-table}
\begin{tabular}{|l|l|}
\cline{1-2}
Battery total capacity                                                 & 0.2 kW              \\ \hline
Battery charging/discharging rate                                      & 0 kW to 0.1 kW                  \\ \hline
Non-schedulable appliances Runtime                                     & 1 (hours) \\ \hline
\end{tabular}
 \vspace*{-3ex}
\end{table}

\subsection{Scheduler Evaluation}
\label{subsec:schedulerEvaluation}
\vspace*{-3ex}
\begin{figure}[tb!]
  \begin{centering}
	\includegraphics[width=0.42\textwidth]{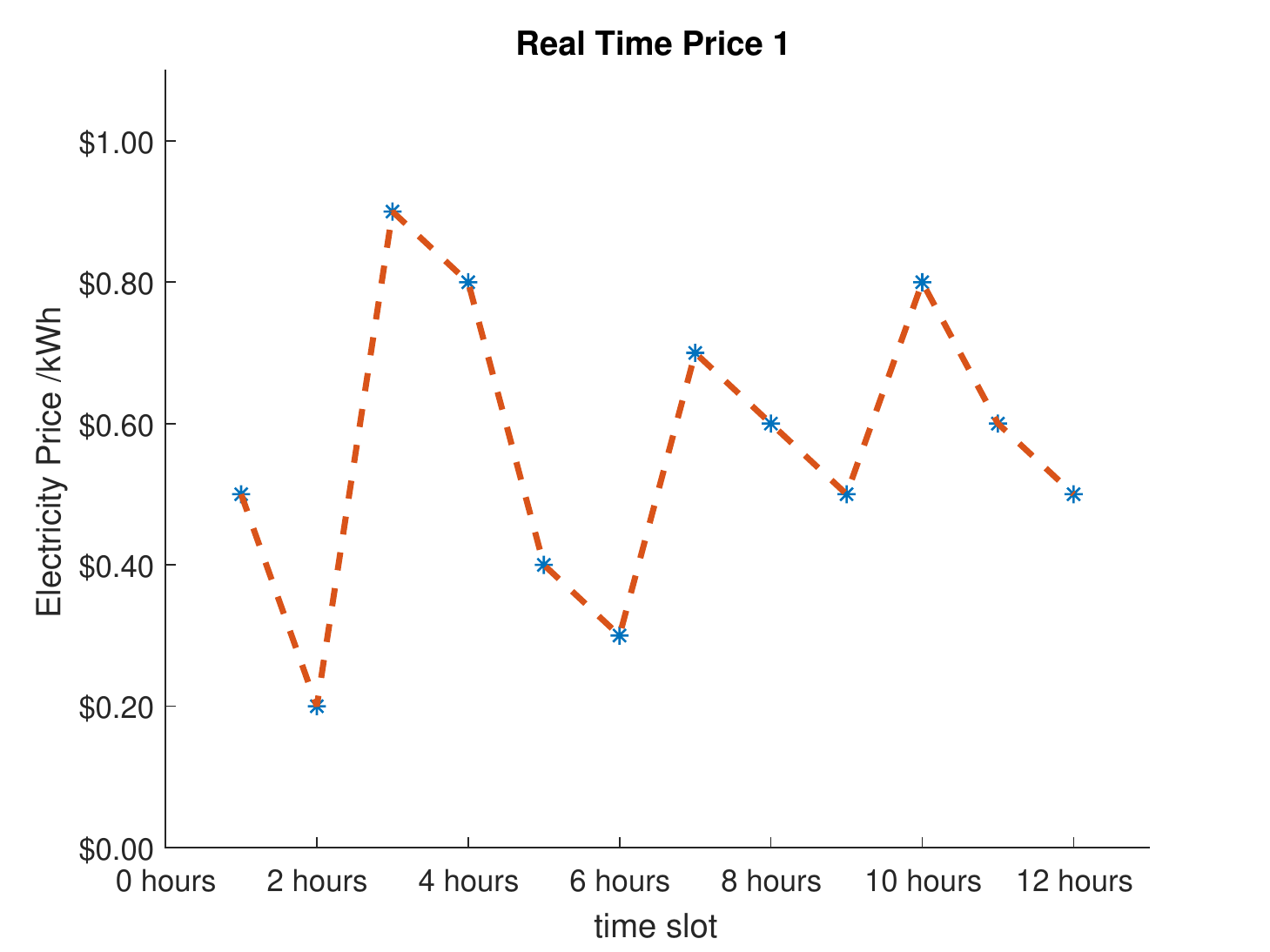}
  \end{centering}
    \vspace*{-3ex}
  \caption{The real time price for performance evaluation.}
  \label{fig:price1}     \vspace*{-3ex}
\end{figure}

\begin{table}[h]
\centering
\caption{Schedulable module evaluation results}
\label{tab:operation-table1}
\begin{tabular}{|l||l|l|l|l|l|l|l|l|l|l|l|l|l}
\cline{1-13}
Time Slot & 1 & 2 & 3 & 4 & 5 & 6 & 7 & 8 & 9 & 10 & 11 & 12              \\ \hline
App 1 & 0 & 0 & 0 & 0 & 0 & 1 & 1 & 1 & 1 & 1 & 0 & 0                  \\ \hline
App 2 & 0 & 0 & 0 & 1 & 1 & 1 & 0 & 0 & 0 & 0 & 0 & 0          \\ \hline
App 3 & 1 & 1 & 0 & 0 & 0 & 0 & 0 & 0 & 0 & 0 & 0 & 0
\\ \hline 
\end{tabular}
 \vspace*{-3ex}
\end{table}

\begin{figure}[tb!]
  \begin{centering}
	\includegraphics[width=0.42\textwidth]{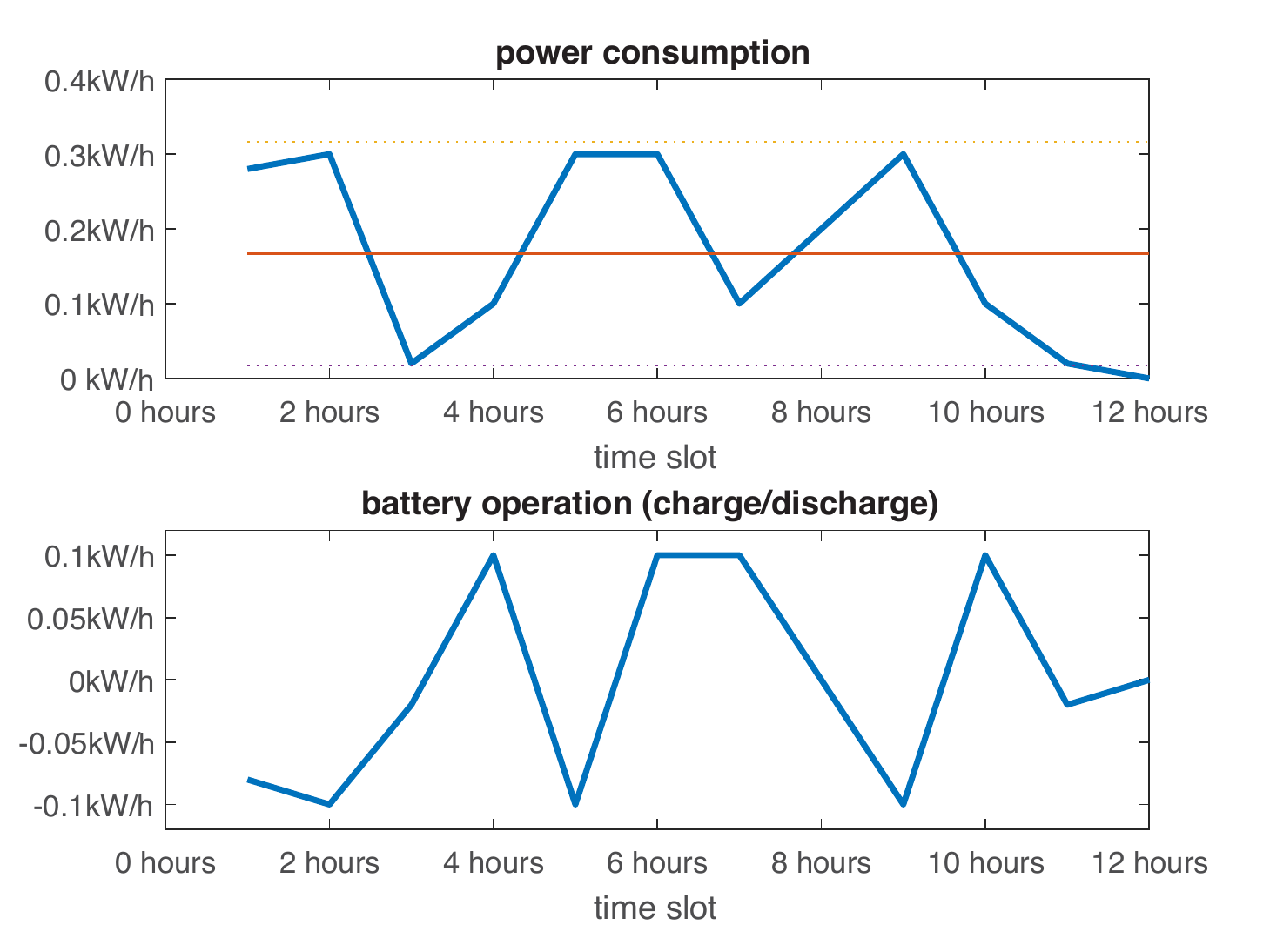}
	\par
  \end{centering}
    \vspace*{-3ex}
  \caption{Scheduling results for schedulable module. }
  \label{fig:sch_3app}     \vspace*{-3ex}
\end{figure}

\begin{table}[h]
\centering
\caption{Non-schedulable module evaluation results}
\label{tab:operation-table2}
\begin{tabular}{|l||l|l|l|l|l|l|l|l|l|l|l|l|l}
\cline{1-13}
Time Slot & 1 & 2 & 3 & 4 & 5 & 6 & 7 & 8 & 9 & 10 & 11 & 12              \\ \hline
App 1 & 0 & 0 & 0 & 0 & 1 & 1 & 1 & 1 & 1 & 0 & 0 & 0                  \\ \hline
App 2 & 0 & 0 & 0 & 0 & 0 & 0 & 0 & 0 & 1 & 1 & 1 & 0          \\ \hline
App 3 & 0 & 1 & 1 & 0 & 0 & 0 & 0 & 0 & 0 & 0 & 0 & 0
\\ \hline 
\end{tabular}
 \vspace*{-3ex}
\end{table}

\begin{figure}[tb!]
  \begin{centering}
	\includegraphics[width=0.42\textwidth]{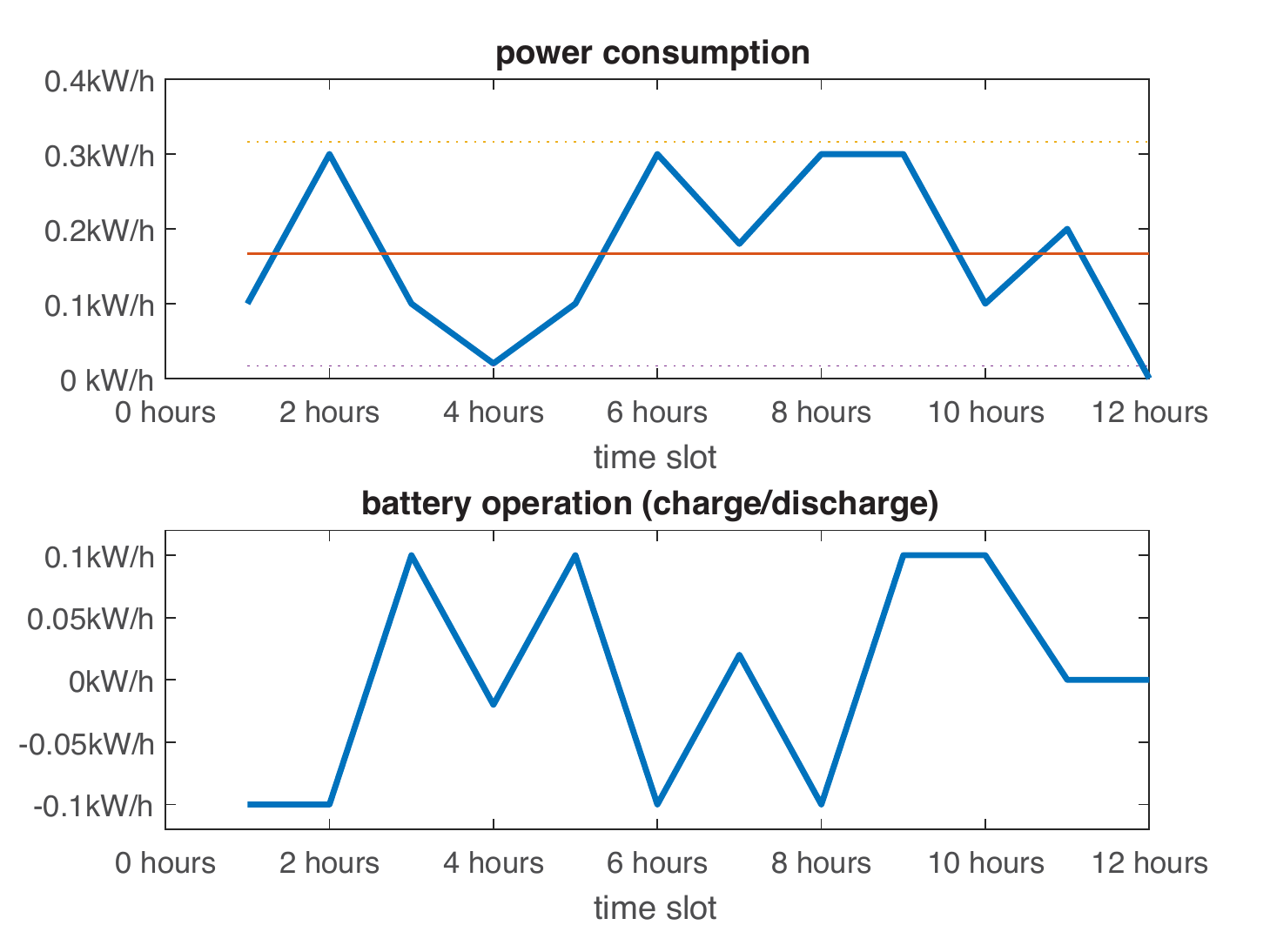}
	\par
  \end{centering}
    \vspace*{-3ex}
  \caption{Scheduling results for non-schedulable module. }
  \label{fig:nonsch_3app}     \vspace*{-4ex}
\end{figure}

Section~\ref{subsec:schedulerEvaluation} will evaluate the system using the setting as shown in Table~\ref{tab:EvaluationSetting-table}. Considering the real-world scenario, the charging/discharging rate of a battery varies in the scheduling horizon. Given the unknown starting time of non-schedulable appliance, the scheduler system defines a certain range of time as the potential starting time of non-schedulable appliance based on its historical usage pattern. In this test, there are three assumptions made. First, each appliance is allowed to run only once and it will stop when its runtime finished. Second, the power consumption of any schedulable appliances within each time slot during the appliances active is a constant value. Third, all appliances have to be finished before the end. 

Fig.~\ref{fig:price1} illustrates the real-time electricity price over 12 hours scheduling horizon. Given the real-time electricity price, online scheduler decides on the actual appliance usage pattern. Table~\ref{tab:operation-table1} and Fig.~\ref{fig:sch_3app} demonstrate the proposed operation of schedulable module and the operation of battery and real time power consumption, respectively. It can be noticed that the power consumption followed the privacy constraint which successfully protect the behavior of all appliance. In the battery operation, the battery will gain energy at the positive value and release energy at negative value in the battery operation as shown in Fig.~\ref{fig:sch_3app}. The non-schedulable module evaluation results are shown in Table~\ref{tab:operation-table2} and Fig.~\ref{fig:nonsch_3app}, respectively. In the Table~\ref{tab:operation-table1} and Table~\ref{tab:operation-table2}, for each appliance, the digits '1' represent that the appliance is switched on and the digits '0' represent that the appliance is switched off. In this scheduling results, the non-schedulable appliance did not run. But the system prepared for it between the first six hours which define as the non-schedulable time zone in this instance. The battery will reserve enough energy for the unpredictable appliance, thus it will follow the privacy constrain. In this module, the system will prepare for the worst case of privacy risk and also find the lowest electricity billing price solution.   

\subsection{Sensitivity analysis of battery capacity and billing price}
\label{subsec:SensiAna}
In this section, we analyze the effect of electricity price and battery capacity on the scheduler the system. The experiment setup of system remains the same as discussed in the Section~\ref{subsec:schedulerEvaluation}. Fig.~\ref{fig:price2}, marked as the second price, is modified from Fig.~\ref{fig:price1}. The minimum electricity price was switched out from the non-schedulable time zone in Fig.~\ref{fig:price2}. To evaluate the effect of price, the test will be repeated with the second price.

The evaluation results is shown in Fig.~\ref{fig:total_billing_price}. The X coordinate represents the battery size and Y coordinate represents the amount of billing price in the U.S. Dollar. As increasing the battery capacity, the billing price converges to a constant value and it indicted the battery capacitor is sufficient in this certain scenario. As shown in Fig.~\ref{fig:total_billing_price}, the system is very sensitive when the battery capacity is relative small. Because the battery is not sufficient to reserve enough energy in a lower price time slot, the billing price drop rapidly when the capacity increase from 8 kW to 20 kW. For the non-schedulable module, it can be noticed that the battery must reserve electricity for the non-schedulable appliance, thus the the total price is slightly higher than the price for a schedulable module. In the second price instance, the minimum price is out of non-schedulable active time range and the system will schedule the battery to reserve energy to prepare for the non-schedulable appliance at higher price time slot. Hence, the total price will increase since the scheduler has to follow the privacy constrain. Therefore, the price drops more when the battery capacity increases in non-schedulable module.

\begin{figure}[tb!]
  \begin{centering}
	\includegraphics[width=0.42\textwidth]{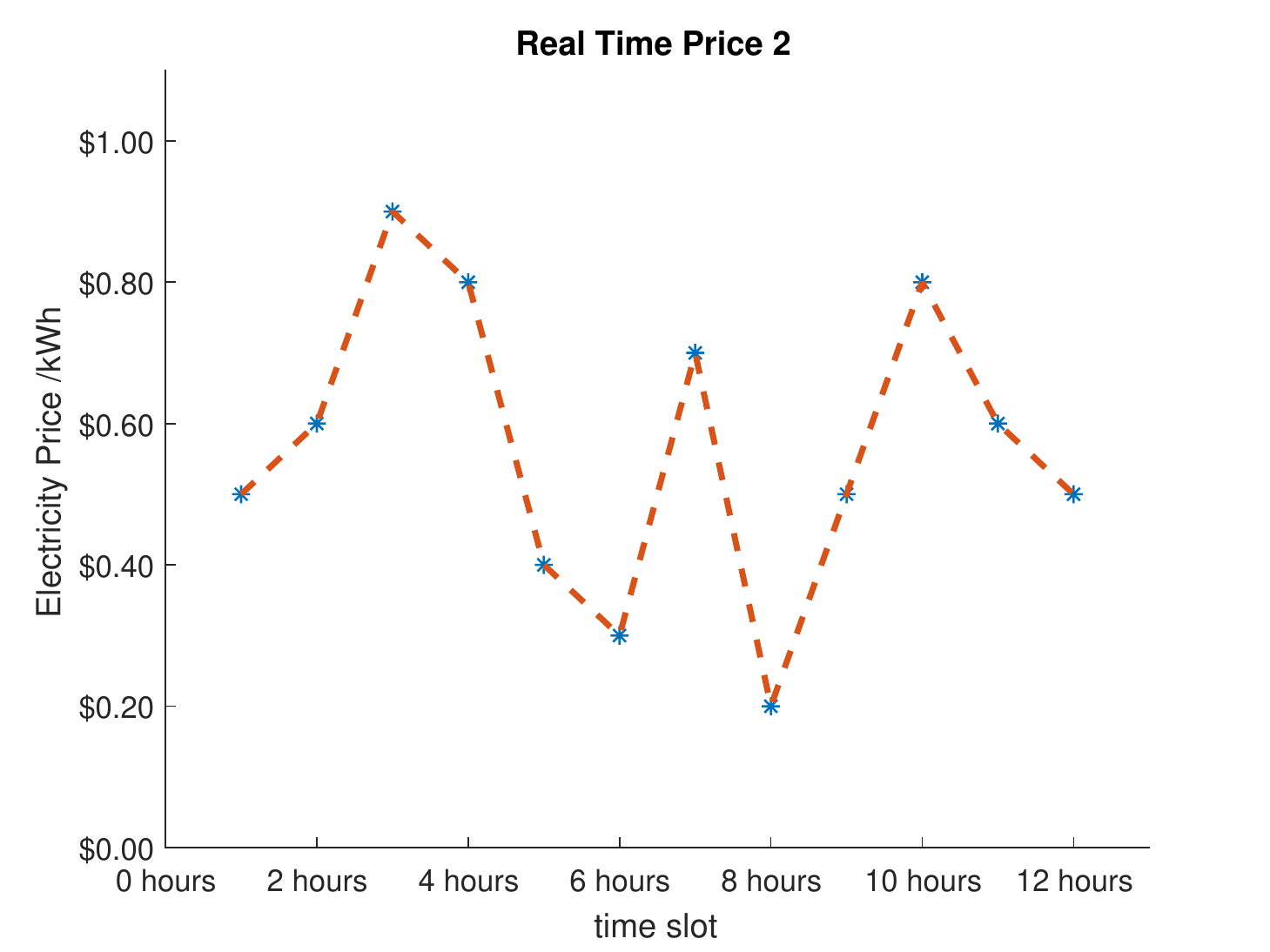}
	\par
  \end{centering}
    \vspace*{-3ex}
  \caption{Modified real time price for sensitivity analysis.}
  \label{fig:price2}     \vspace*{-3ex}
\end{figure}

\begin{figure}[tb!]
  \begin{centering}
	\includegraphics[width=0.42\textwidth]{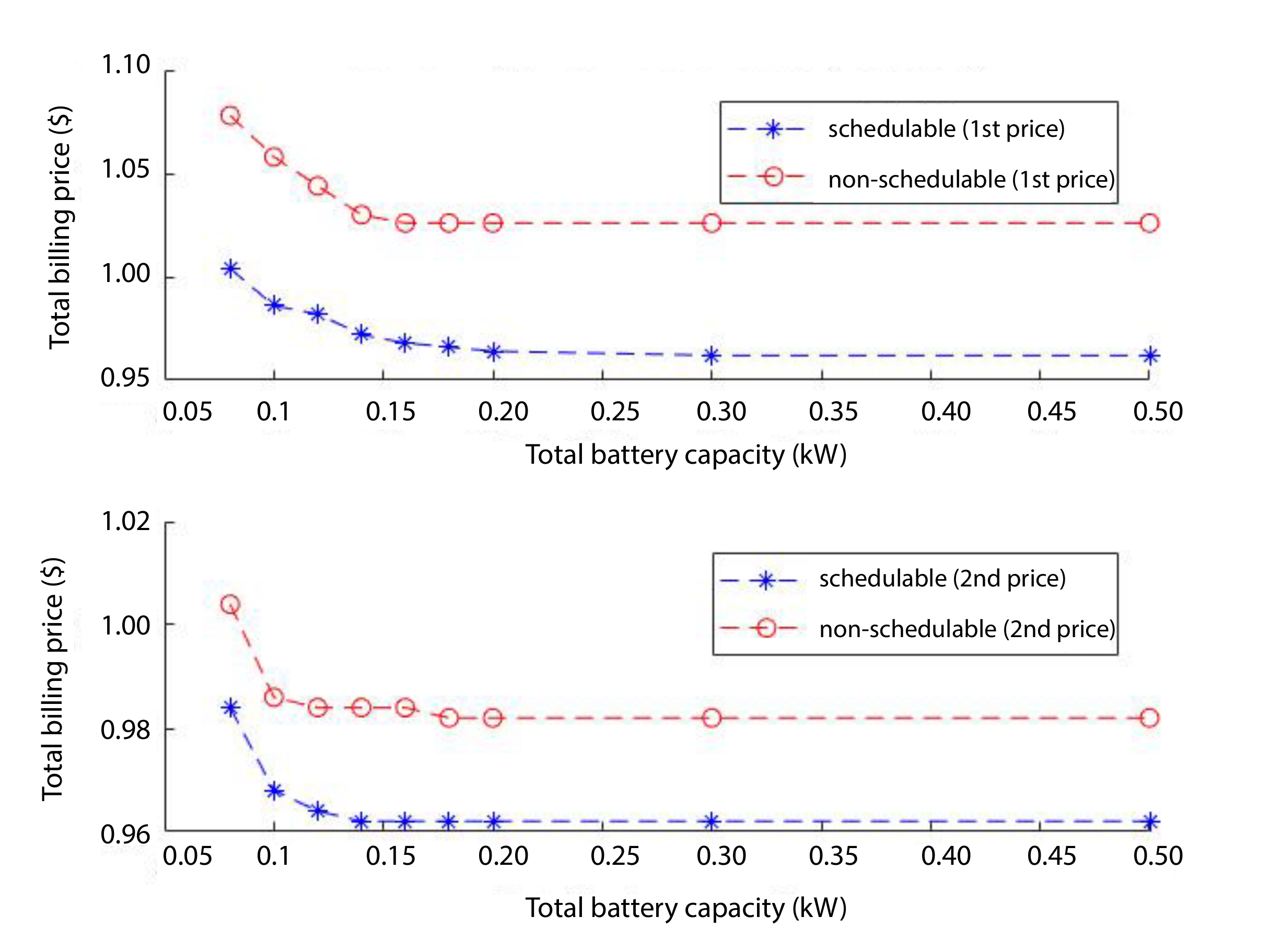}
	\par
  \end{centering}
    \vspace*{-3ex}
  \caption{Total billing price VS total battery capacity sizes.  }
  \label{fig:total_billing_price}     \vspace*{-4ex}
\end{figure}

The scheduler is able to protect the appliance behavior information and also able to obtain a better price solution. The scheduler is build with high flexibility and it can handle more complex scenarios, such as flexible battery charging/discharging rate, a smart home system with a large number of schedulable appliances and non-schedulable appliances.

\section{Summary and Future Work}\label{sec:conclusion}

Smart homes promise many potentials but also raise new privacy concerns. 
This paper considers the effects of fake guideline electricity price and 
non-schedulable appliances' operation uncertainties in 
appliance scheduling for smart homes. Different from existing research, this work aims to not only minimize electricity cost but also protect customers' privacy. The proposed framework, \textbf{PACES},  is evaluated using publicly
released households' data sets. Our experimental study shows that \textbf{PACES} can effectively protect customers' privacy and
satisfy their immediate service requirement with a small increase in electricity cost. \textbf{PACES} can be somewhat time consuming if a household has a large number of appliances. There are great research opportunities in the area of privacy protection and cost reduction for smart homes.




{\small
\bibliographystyle{ieee}
\bibliography{reference}
}

\end{document}